\documentclass[10pt,journal]{IEEEtran}

\usepackage{amsthm}
\usepackage{amsmath}
\usepackage{amssymb}
\usepackage{amsmath}
\usepackage{amssymb}
\usepackage{epsfig}
\usepackage{epsf}
\usepackage{subfigure}
\usepackage{graphicx}
\usepackage{url}
\usepackage[]{authblk}
\usepackage{cite}
\usepackage{epstopdf}

\def\BibTeX{{\rm B\kern-.05em{\sc i\kern-.025em b}\kern-.08em
    T\kern-.1667em\lower.7ex\hbox{E}\kern-.125emX}}

\newtheorem{deft}{Definition}
\newtheorem{thm}{Theorem}
\newtheorem{lem}{Lemma}
\newtheorem{crl}{Corollary}

\newcounter{numcount}

\begin{document}
\newcommand{\caseone} { { \nearrow } { \hspace{-3.8mm} \searrow } {\hspace{-3.85mm}  \raisebox{4.7pt}{{$\rightharpoonup$}}} {\hspace{-3.85mm}  \raisebox{-4.7pt}{{$\rightharpoondown$}}} }

\newcommand{\casetwo} { { \searrow } {\hspace{-3.85mm}  \raisebox{4.7pt}{{$\rightarrow$}}} {\hspace{-3.85mm}  \raisebox{-4.7pt}{{$\rightharpoondown$}}} }

\newcommand{\casethree} { { \nearrow } {\hspace{-3.85mm}  \raisebox{4.7pt}{{$\rightarrow$}}} {\hspace{-3.85mm}  \raisebox{-4.7pt}{{$\rightharpoondown$}}} }

\newcommand{\casefour} { { \raisebox{4.7pt}{{$\rightarrow$}}} {\hspace{-3.85mm}  \raisebox{-4.7pt}{{$\rightarrow$}}} }

\newcommand{\casefive} { { \raisebox{4.7pt}{{$\rightarrow$}}} }

\newcommand{\casesix} { { \searrow } {\hspace{-3.85mm}  \raisebox{4.7pt}{{$\rightarrow$}}}  }

\newcommand{\caseseven} { { \nearrow } {\hspace{-3.85mm}  \raisebox{4.7pt}{{$\rightharpoonup$}}} }

\newcommand{\caseeight} { { \nearrow } { \hspace{-3.8mm} \searrow } {\hspace{-3.85mm}  \raisebox{4.7pt}{{$\rightharpoonup$}}} }

\newcommand{\casefifteen} { { \nearrow } {\hspace{-3.85mm} \searrow } }

\title{Epidemic Threshold of an SIS Model in Dynamic Switching Networks}

\author{Mohammad~Reza~Sanatkar,
        Warren~N.~White,
        Balasubramaniam~Natarajan,
        Caterina~M.~Scoglio,
        and~Karen~A.~Garrett
        \thanks{Mohammad Reza Sanatkar is with the Department of Electrical and Computer Engineering, Duke University, Durham, NC, USA. Email: {\sffamily reza.sanatkar@duke.edu}.}
        \thanks{Warren N. White is with the Department of Mechanical and Nuclear Engineering, Kansas State University, Manhattan, KS, USA. Email: {\sffamily wnw@k-state.edu}.}
        \thanks{Balasubramaniam Natarajan and Caterina M. Scoglio are with the Department of Electrical and Computer Engineering, Kansas State University, Manhattan, KS, USA. Email: {\sffamily bala@k-state.edu; caterina@k-state.edu}.}
        \thanks{Karen A. Garrett is with the Institute for Sustainable Food Systems and Plant Pathology Department, University of Florida, Gainesville, FL, USA. Email: {\sffamily karengarrett@ufl.edu}.}
}

\maketitle

\begin{abstract}
In this paper, we analyze dynamic switching networks, wherein the networks switch arbitrarily among a set of topologies. For this class of dynamic networks, we derive an epidemic threshold, considering the SIS epidemic model. First, an epidemic probabilistic model is developed assuming independence between states of nodes. We identify the conditions under which the epidemic dies out by linearizing the underlying dynamical system and analyzing its asymptotic stability around the origin. The concept of joint spectral radius is then used to derive the epidemic threshold, which is later validated using several networks (Watts-Strogatz, Barabasi-Albert, MIT reality mining graphs, Regular, and Gilbert). A simplified version of the epidemic threshold is proposed for undirected networks. Moreover, in the case of static networks, the derived epidemic threshold is shown to match conventional analytical results. Then, analytical results for the epidemic threshold of dynamic networks are proved to be applicable to periodic networks. For dynamic regular networks, we demonstrate that the epidemic threshold is identical to the epidemic threshold for static regular networks. An upper bound for the epidemic spread probability in dynamic Gilbert networks is also derived and verified using simulation. 
\end{abstract}


\vspace{0.1 in}

\section{Introduction\label{section:introduction}}
Epidemics typically start with some initial infected nodes. There is a probability that healthy neighboring nodes, close to the infected nodes, become infected. With time and, in some cases, with external intervention, infected nodes recover and revert to a healthy state. The study of epidemic dispersal on networks aims at understanding how epidemics evolve and spread in networks. When an infection enters a network, it is very useful to be able to determine whether it will die out or become a massive out break. The epidemic threshold addresses this question, taking into account both the network topology and the epidemic strength. The spread of epidemics in static networks has been studied extensively \cite{Kostova,Faloutsos3,Faloutsos4,Chen,Piet,Ganesh,Lindquist,Vespignani2}. More recently, epidemic dispersal in dynamic networks has garnered attention \cite{Faloutsos1,Faloutsos2,Volz,Taylor,Gross,Rocha, Bansal, Smieszek,Kamp,Rocha2,Volz2,Chen5}. In a dynamic network, links between nodes are functions of time. This dynamic nature presents a more realistic picture of the spread of epidemic infection. Let us examine a few examples of dynamic networks which can be modeled. First in the area of human epidemics, the nodes (human beings) are constantly moving from one location to another, thus the contact networks between people change over time. In the case of a mobile ad hoc network, each device can move independently and change its links to other devices frequently, which results in dynamic topologies. Many Bluetooth devices are becoming susceptible to viruses such as Cabir or Comm Warrior. Another classic example of dynamic networks is the networks for the spread of diseases among animals or plants where the factors that influence the spread of disease-causing spores are typically dynamic \cite{karen1, karen2, karen3, karen4}.\\
\indent In this paper, we consider the Susceptible-Infected-Susceptible (SIS) model for epidemic spread. In this model, healthy nodes can become infected through infected neighbors; infected nodes have a probability of recovery. In the SIS model, an infected node after recovery becomes susceptible to infection again. We assume that infected nodes have the recovery probability $\delta$, and that the probability of an infected node infecting its healthy neighbor is $\beta$. There are numerous examples in the real world for which the SIS model is the right choice to model epidemics. For example, several diseases caused by bacteria do not produce immune response in the body, thereby allowing the recovered individuals to return to the susceptible population. The SIS model can also be used for opinion analysis. In a contact network of people, every person can have either a positive or negative attitude toward a subject and can constantly change from positive to negative and vice versa. \\
\indent In this paper, we study dynamic switching networks in which adjacency matrices are randomly chosen from sets of matrices at each step. We do not consider any temporal correlation between subsequent adjacency matrices. In other words, we assume that the process of choosing the adjacency matrix at time index $t$ is independent from adjacency matrices chosen at previous time indices. However, this assumption does not hold for dynamic networks with strong temporal correlations between successive adjacency matrices.\\
\indent  First, the nonlinear dynamic nature of nodal infection probabilities, assuming independence among the states of nodes, will be developed. Then, we will prove that the origin is always an equilibrium point of this time-varying dynamical system, and its stability depends on  network topology and values of $\delta$ and $\beta$. Next, the linearized version of the nonlinear epidemic system is derived to determine asymptotic stability of the origin. We show that if the origin is not a stable equilibrium of the system, the epidemic spreads. The joint spectral radius of a set of matrices is also defined. In Theorem \ref{thm:JSR}, we employ the concept of joint spectral radius in order to derive the analytical epidemic threshold for dynamic networks. In Theorem \ref{thm:symmetric}, the simplified version of the epidemic threshold for undirected networks is derived. Because the epidemic threshold for undirected networks depends only on the largest spectral radius of a set of system matrices, evaluation of the epidemic threshold is computationally less expensive compared to directed networks. In Corollary \ref{crl:Static}, the derived epidemic threshold confirms conventional analytical results for static networks, and then the proposed epidemic threshold for dynamic networks is extended to periodic networks. We also study epidemic spread in dynamic regular networks and show that the epidemic threshold for dynamic regular networks is identical to that for static regular networks. An upper bound for the probability of an epidemic spreading in dynamic Gilbert networks is derived. Finally, we simulate epidemics in Watts-Strogatz, Barabasi-Albert, Regular, and dynamic Gilbert networks in order to validate our analytical results. Additionally, we examine our theoretical results in the context of real networks by considering MIT reality mining graphs \cite{MIT}.\\
\indent The rest of the paper is organized as follows. In Section \ref{section:relatedwork}, we review related prior work on epidemic thresholds in dynamic networks. Section \ref{section:analytical} contains general analytical results for epidemic thresholds for dynamic networks as well as simplified epidemic thresholds for special cases of dynamic networks. In Section \ref{section:simulation}, we use simulation results to validate our theoretical analysis.

\section{Related Previous Work\label{section:relatedwork}}
In \cite{Volz}, the epidemic threshold for an Susceptible-Infected-Recovered (SIR) model is derived for a simple class of dynamic random networks. In these networks, the number of neighbors of a given node is fixed, but its neighbors change stochastically as a Poisson process through instantaneous neighbor exchanges. Pairs of edges are chosen continually and randomly with equal probability, and they are instantaneously interchanged. In \cite{Taylor}, the authors present a model that describes an SIS epidemic on dynamic networks using a set of ordinary differential equations. The SIS effective degree model for a static contact network in \cite{Lindquist} is modified by introducing link activation and deletion rates. The epidemic threshold is calculated for this model and shows that the limiting of the maximum nodal degree of a network can prevent the outbreak of epidemic.
In \cite{Gross}, the authors study epidemic dynamics on an adaptive network, in which susceptible nodes try to avoid infection from infected nodes. To achieve this, the susceptible nodes cut their links with infected nodes, using a constant rewiring probability, and replace them by links with other susceptible nodes. The adaptive rewiring increases isolation of infected individuals, and simultaneously, contributes to the formation of a highly connected susceptible cluster. Consequently, the local effect of rewiring increases the epidemic threshold, while the topological effect renders the network vulnerable to epidemics. Moreover, the adaptive nature of the system leads to the emergence of bistability and limit cycles in its dynamical behavior; however, only one continuous dynamical transition exists in the static networks. 
In \cite{Faloutsos1}, the authors derive the epidemic threshold for dynamic networks with alternating (periodic) adjacency matrices. They consider the SIS model for epidemic propagation in networks and show that if the dynamic behavior of a time-varying network can be characterized by T repeating alternating graphs, and $L=\{\boldsymbol{A_{1}}, \boldsymbol{ A_{2}}, ...,\boldsymbol{ A_{T}}\}$, then the system matrix, $\boldsymbol{S}$, of this dynamical system can be expressed as

\begin{equation}
\boldsymbol{ S}=\prod_{i=1}^{T}{\left[(1-\delta)\boldsymbol{ I}+\beta \boldsymbol{ A_{i}})\right]} \label{eq:white}\text{,}
\end{equation} 

\noindent where the dimension of $\boldsymbol{ A_{i}}$ is $n \times n$ ($n$ is the number of nodes), ${\bf I}$ is an $n$ by $n$ identity matrix, and $\delta$ and $\beta$ denote, respectively, the recovery probability and infection probability. The authors of \cite{Faloutsos1} prove that if the spectral radius of the system matrix is less than $1$, the origin is an asymptotically stable equilibrium point of the system, and the epidemic dies out. This result holds only for cases having repeating patterns of adjacency matrices, with the order of repetition is preserved. In \cite{Faloutsos2}, the authors study malware propagation on mobile ad hoc networks. They extend their results for the epidemic threshold of periodic networks in \cite{Faloutsos1} to general cases in which the repeating order of adjacency matrices can be arbitrary. In Theorem I (Mobility model threshold) of their paper, they state that if a mobility model can be represented as a sequence of connectivity graphs $L=\{\boldsymbol{A_{1}},\boldsymbol{ A_{2}},...,\boldsymbol{ A_{T}}\}$, with one adjacency matrix $\boldsymbol{A_{t}}$ for each index $t\in\{1,2,..,T\}$, then the epidemic threshold is

\begin{equation}
\tau=\lambda_{S}\text{,}
\end{equation} 

\noindent where $\lambda_{S}$ is the largest eigenvalue of the matrix $\boldsymbol{S}$ defined in Eq. (\ref{eq:white}).\\
\indent This theorem claims that, for a given dynamic network with an adjacency matrix arbitrarily chosen from a set of matrices at each index, the condition for asymptotic stability is that the spectral radius of the matrix $\boldsymbol{ S}$ is less than $1$. This is different from our analytical results in Theorem \ref{thm:JSR} discussed in Section \ref{section:analytical} given the same assumptions. 

\vspace{0.1 in}

\section{ANALYTICAL RESULTS\label{section:analytical}}
In this section, we develop a dynamical system for epidemic spread, assuming spatial independence between states of nodes in a given network. A linearized version of the dynamical system is then derived to determine the epidemic threshold. Next, by employing the joint spectral radius, we quantify the epidemic threshold for dynamic networks, proving that the epidemic threshold in undirected networks depends only on the maximum spectral radius of the set of system matrices. Then we extend the results so obtained to static and periodic networks. We also show that the epidemic threshold for dynamic regular networks is equal to the epidemic threshold of static regular networks. Finally, we calculate the upper bound for the probability of epidemic spreading in dynamic Gilbert networks.\\
\indent Using the above cited assumption, we can write the infection probability of each node in the network as 

\begin{equation}
p_{i}(t+1)=1-p_{i}(t)\delta-(1-p_{i}(t))\prod_{j\in N_{i}(t)}{\left[1-p_{j}(t)\beta\right]}\label{eq:P}\text{,}
\end{equation}
\noindent where $N_{i}(t)$ denotes the set of neighbors to node $i$ at index $t$, which is a function of time. Infection probabilities of nodes can be interpreted as state variables of a dynamical system. Eq. (\ref{eq:P}) shows that infection probabilities at a given index are nonlinear functions of infection probabilities of the previous index. Therefore, the epidemic dynamical system is nonlinear. The corresponding state space of this nonlinear system is the subspace $[0,1]^{n}$ in $R^{n}$, where $n$ is the number of nodes in the network. For instance, when $n=2$, the state space is a rectangle with vertices represented by points of $(0,0)$, $(0,1)$, $(1,0)$, and $(1,1)$ in $R^{2}$. 
Given initial infection probabilities of the nodes, we can calculate the evolving trajectory of infection probabilities in state space. The family of evolving trajectories of states in the state space is called the phase portrait. The study of the steady state behavior of dynamical systems requires finding the equilibrium points. If $P^{*}$ is an equilibrium point,  $P^{*}(t+1)=P^{*}(t)=P^{*}$. Therefore, we can write

\begin{equation}
p^{*}_{i}=1-p^{*}_{i}\delta-(1-p^{*}_{i})\prod_{j\in N_{i}(t)}{\left[1-p^{*}_{j}\beta\right]}\label{eq:PE}\text{,}
\end{equation}

\noindent where $p_{i}^{*}$ is the infection probability of the $i^{th}$ node in steady state, if an asymptotic equilibrium point is present. Eq. (\ref{eq:PE}) is the equilibrium equation corresponding to node $i$. In order to find equilibrium points of a given epidemic system with $n$ nodes, we must solve a system of $n$ equations with $n$ unknowns. In the case of dynamic networks, this system of equations has been changing with time, and equilibrium points, by definition, are static points that satisfy this system of equations for all time. An epidemic dynamical system may have more than one equilibrium point. 
From Eq. (\ref{eq:PE}), the origin is always an equilibrium point, meaning that, for all values of $\beta$ and $\delta$ and for any arbitrary topologies of network, the origin is always an equilibrium point. However, stability of the origin depends on $\beta$, $\delta$, and underlying topology of the network. If the origin is an asymptotically stable equilibrium point, the epidemic dies out. However, the epidemic spreads when the origin is an unstable equilibrium point. 

\subsection{Linearization of System Equations}
\begin{figure}[]
\begin{center}
\includegraphics[width=3.5in]{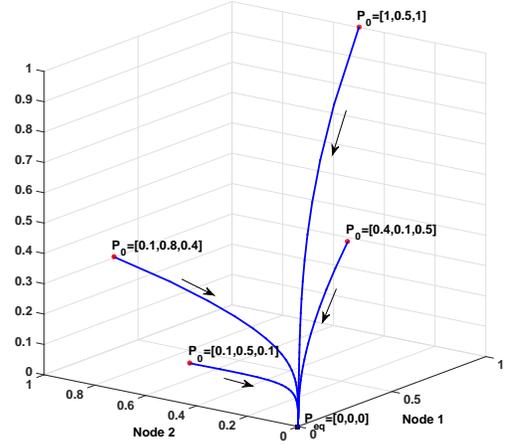}
\caption{Phase portrait of an epidemic network with three nodes in which the epidemic dies out.}
\label{fig:die}
\end{center}
\vspace{-0.4cm}
\end{figure}

One way to identify the stability status of an equilibrium point of a nonlinear system is to study stability of the linearized system at that equilibrium point. In the case of epidemic networks, we are interested in determining the stability status of the origin. If the origin is an asymptotically stable point, the epidemic dies out only if no other equilibrium points exist in the subspace $[0,1]^{n}$; otherwise, asymptotic stability is only local. Therefore, we linearize the epidemic nonlinear system at the origin. Ignoring the nonlinear terms in Eq. (\ref{eq:P}), we can write

\begin{equation}
p_{i}(t+1)=p_{i}(t)(1-\delta)+\sum_{j\in N_{i}(t)}{p_{j}(t)\beta}\label{eq:PL}\text{.}
\end{equation}
  
\noindent We can rewrite (\ref{eq:PL}) in the form of a matrix equation as

\begin{equation}
\boldsymbol{P_{t+1}}=\left[(1-\delta)\boldsymbol{I}+\beta \boldsymbol{A_{t}}\right]\boldsymbol{P_{t}}\label{eq:PM}\text{,}
\end{equation}

\noindent where $\boldsymbol{P_{t}}=[p_{1}(t) p_{2}(t) ... p_{n}(t)]^{T}$ is the system state, $\boldsymbol{A_{t}}$ is the adjacency matrix at index $t$, and $I$ denotes the $n\times n$ identity matrix. We define $\boldsymbol{M_{t}}$ as $\boldsymbol{M_{t}}=\left[(1-\delta)\boldsymbol{I}+\beta \boldsymbol{A_{t}}\right]$ the system matrix at index $t$. 

\begin{figure}[]
\begin{center}
\includegraphics[width=3.5in]{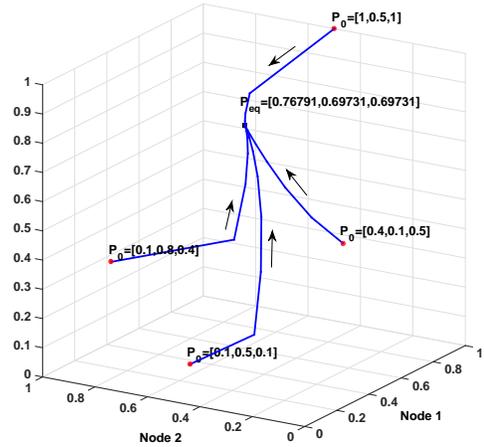}
\caption{Phase portrait of an epidemic network with three nodes in which the epidemic spreads.}
\label{fig:spread}
\end{center}
\vspace{-0.4cm}
\end{figure}

\subsection{Epidemic Threshold}
As mentioned earlier, stability of the origin in the epidemic dynamical system determines under what conditions epidemics die out. 
\noindent If the origin is an asymptotically stable equilibrium point, the system state reaches the origin and infection probabilities of all nodes become zero and remain zero. Fig. \ref{fig:die} depicts trajectories of state evolution of an epidemic network for different initial infection probabilities. The network has three nodes and its adjacency matrix is static. Node $1$ is connected to nodes $2$ and $3$, but no link exists between nodes $2$ and $3$. Each axis of the 3D plot in Fig. \ref{fig:die} represents the infection probability of one of the three nodes over time. Dots with $\boldsymbol{P_{0}}$ label in state space represent initial infection probabilities, and the dot with $\boldsymbol{P_{eq}}$ represents the equilibrium point of the system. The probability constants used are $\delta=0.2$ and $\beta=0.1$. We see that all the trajectories reach the origin regardless of their initial states. Fig. \ref{fig:die} shows that, for these probability constants, the origin is asymptotically stable and epidemics die out. 
In the case of spreading epidemics, the origin is not a stable equilibrium point and the state variables converge on a non-zero equilibrium point and remain at that point. This equilibrium point determines the final fraction of infected nodes. Fig. \ref{fig:spread} shows trajectories of state evolution of the same network depicted in Fig. \ref{fig:die}. The initial states are identical to initial states in Fig. \ref{fig:die}. However, the value of $\beta$ is different. In Fig. \ref{fig:spread}, $\beta=0.6$, and the epidemic spreads out. For all different infection probability vectors, the system state reaches the equilibrium point $\boldsymbol{P_{eq}}=[0.76791,0.69731,0.69731]$ with elements that represent, respectively, infection probability of nodes $1$, $2$, and $3$ in steady state.\\
\indent Before tackling the problem of stability of the origin, essentials definitions must be presented: 

\begin{deft} Given $M$ is a set of matrices, define
\label{def:JSR}
\begin{equation}
\widehat{\rho}_{k}(M,||.||):=sup\left \{\left|\left|\prod_{i=1}^{k}{\boldsymbol{M_{i}}}\right|\right|:\boldsymbol{M_{i}}\in M \text{ for } 1\leq i\leq  k \right\} \nonumber \text{,}
\end{equation}

\noindent where $\widehat{\rho}_{k}(M)$ is the largest possible norm of all products of $k$ matrices chosen in the set $M$. The joint spectral radius $\widehat{\rho}(M)$ is defined as \normalfont{\cite{Strang}}

\begin{equation}
\widehat{\rho}(M):=\lim_{k\to\infty}\widehat{\rho}_{k}(M,||.||)^{\frac{1}{k}} \text{.}
\end{equation}

\end{deft}

\noindent Therefore, the joint spectral radius of set $M$ is the maximum possible norm of products of matrices in the set $M$ when the number of products $k$ goes to infinity.
\begin{deft} Given $M$ is a set of matrices, define
\begin{equation}
\overline{\rho}_{k}(M):=sup\left \{\rho\left(\prod_{i=1}^{k}{\boldsymbol{M_{i}}}\right):\boldsymbol{M_{i}}\in M \text{ for } 1\leq i\leq  k \right\} \nonumber \text{,}
\end{equation}

\noindent where $\rho$ denotes the spectral radius and $\overline{\rho}_{k}(M)$ is the largest possible spectral radius of all products of $k$ matrices chosen in the set $M$. The generalized spectral radius $\overline{\rho}(M)$ is defined as \normalfont{\cite{Lagarias}}

\begin{equation}
\overline{\rho}(M):=\lim_{k\to\infty}\overline{\rho}_{k}(M)^{\frac{1}{k}} \text{.}
\end{equation}

\end{deft}

\noindent In \cite{Wang}, the authors prove that, for a bounded set of matrices, the generalized spectral radius is equal to the joint spectral radius.
\begin{lem} \label{lem:four}\textbf{Four-member inequality}.\normalfont{\cite{Lagarias2}} For a given arbitrary set of matrices $M$ and any $k \geq 1$
\begin{equation}
\overline{\rho}_{k}(M)^{\frac{1}{k}}\leq \overline{\rho}(M) \leq \widehat{\rho}(M) \leq \widehat{\rho}_{k}(M)^{\frac{1}{k}}\nonumber \text{,}
\end{equation}
independent of the induced norm used to define $\widehat{\rho}_{k}(M)$.
\end{lem}
\indent Let us consider a set $L$ of all possible adjacency matrices $A_{i}$ and at each time point the adjacency matrix is randomly chosen from this set. $L$ is surely bounded and may be finite or infinite. We define $M$ as the set of system matrices corresponding to the adjacency matrices in $L$. $\boldsymbol{M_{i}}$ is a member of the set $M$ and defined as $\boldsymbol{M_{i}}=\left[(1-\delta)\boldsymbol{I}+\beta \boldsymbol{A_{i}}\right]$. Therefore, $M$ is also bounded. If $L$ is finite, $M$ is also finite; if $M$ is infinite, $L$ is also infinite.

\begin{thm}{\label{thm:JSR}Consider a set $L$ of all possible adjacency matrices of a dynamic network, with infection probability $\beta$, and recovery probability $\delta$. If the joint spectral radius of set $M$ of system matrices is less than one, the origin is an asymptotically stable equilibrium point and the epidemic dies out.\\ } 
\end{thm}

\begin{proof}
Assume $\widehat{\rho}(M)=l<1$. Definition \ref{def:JSR}, a value $k_{0}$ exists for each $\epsilon>0$ with $\epsilon<\min(l,1-l)$ such that

\begin{equation}
l-\epsilon<\widehat{\rho}_{k}(M)^{\frac{1}{k}} < l+\epsilon \text{,  }\forall k>k_{0} \nonumber \text{.}
\end{equation}

\noindent If we raise all sides of the above inequality to the power $k$, we can conclude that 
\begin{equation}
\lim_{k\to\infty}\widehat{\rho}_{k}(M)=0 \nonumber \text{.}
\end{equation}
\noindent  Considering the formula of $\widehat{\rho}_{k}(M)$ in Definition \ref{def:JSR}, for any product of matrices $\boldsymbol{M_{i}} \in M$, we can write

\begin{equation}
0\leq\left|\left|\prod_{i=1}^{k}{\boldsymbol{M_{i}}}\right|\right|\leq \widehat{\rho}_{k}(M) \nonumber \text{.}
\end{equation}

\noindent We have shown that the right hand side of the above inequality goes to zero when $k\rightarrow \infty$. Therefore, for any product of $\boldsymbol{M_{i}}$s, we can write 

\begin{equation}
\lim_{k\to\infty}\left|\left|\prod_{i=1}^{k}{\boldsymbol{M_{i}}}\right|\right|=0 \nonumber \text{.}
\end{equation}

\noindent Considering the fact that $||\boldsymbol{A}||=0 \Leftrightarrow \boldsymbol{A}=0$, if $\underset{k\to\infty}{\lim}\left|\left|\prod_{i=1}^{k}{\boldsymbol{M_{i}}}\right|\right|=0$, we can write

\begin{equation}
\lim_{k\to\infty}\prod_{i=1}^{k}{\boldsymbol{M_{i}}}=0 \nonumber \text{.}
\end{equation}

\noindent If $\underset{k\to\infty}{\lim}\prod_{i=1}^{k}{\boldsymbol{M_{i}}}=0$, for any initial infection probability vector $\boldsymbol{P_{0}}$, we can write

\begin{equation}
\lim_{k\to\infty}\left[\prod_{i=1}^{k}{\boldsymbol{M_{i}}}\right] \boldsymbol{P_{0}}=\lim_{k\to\infty}\left[\prod_{i=1}^{k}{\left[(1-\delta)\boldsymbol{I}+\beta \boldsymbol{A_{i}}\right]}\right] \boldsymbol{P_{0}}=0 \nonumber \text{,}
\end{equation}

\noindent which shows that the origin is an asymptotically stable equilibrium point for any initial infection probability vector and any random sequence of adjacency matrices if the joint spectral radius is less than $1$. In this case, the final infection probability vector is zero and the epidemic dies out.  
\end{proof}

\begin{thm}{\label{thm:symmetric}Consider a set $L$ of all possible adjacency matrices of a dynamic network with undirected graphs, set $M$ of the system matrices corresponding to set $L$ of the adjacency matrices, infection probability $\beta$, and recovery probability $\delta$. If the largest spectral radius of the matrices in set $M$ is less than $1$, the origin is an asymptotically stable equilibrium point and the epidemic dies out.\\ } 
\end{thm}

\begin{proof}
If the network graph is undirected, its corresponding adjacency matrix is symmetric so that its corresponding system matrix is also symmetric. We know that for a given symmetric matrix $\boldsymbol{M_{i}}$, we can calculate the induced $2$ norm of $\boldsymbol{M_{i}}$ as 
\begin{equation}
||\boldsymbol{M_{i}}||_{2}=\sqrt{\rho(\boldsymbol{M_{i}}^{T}\boldsymbol{M_{i}})}=\sqrt{\rho(\boldsymbol{M_{i}})^{2}}=\rho(\boldsymbol{M_{i}}) \nonumber \text{.}
\end{equation}

\noindent If we use the induced $2$ norm to calculate $\widehat{\rho}_{1}(M)$ and $\overline{\rho}_{1}(M)$, we can write, respectively,

\begin{equation}
\widehat{\rho}_{1}(M)=sup \left\{\left|\left|\boldsymbol{M_{i}}\right|\right|_{2}:\boldsymbol{M_{i}}\in M\right\}=sup \left\{\rho(\boldsymbol{M_{i}}):\boldsymbol{M_{i}}\in M\right\}\nonumber \text{}
\end{equation}

\noindent and 

\begin{equation}
\overline{\rho}_{1}(M)=sup \left\{\rho(\boldsymbol{M_{i}}):\boldsymbol{M_{i}}\in M\right\}\nonumber \text{.}
\end{equation}

\noindent Therefore, we conclude that for a set of symmetric matrices,

\begin{equation}
\widehat{\rho}_{1}(M)=\overline{\rho}_{1}(M)=sup \left\{\rho(\boldsymbol{M_{i}}):\boldsymbol{M_{i}}\in M\right\} \label{eq:sup} \text{.}
\end{equation}

\noindent Moreover, we mentioned in Lemma \ref{lem:four} that the four-member inequality holds for any $k \geq 1$. Therefore, we can write

\begin{equation}
\overline{\rho}_{1}(M)\leq \overline{\rho}(M) \leq \widehat{\rho}(M) \leq \widehat{\rho}_{1}(M)\label{eq:fourone} \text{.}
\end{equation}

\noindent Considering Eqs. (\ref{eq:sup}) and (\ref{eq:fourone}), we can write

\begin{equation}
\overline{\rho}(M)=\widehat{\rho}(M)=sup\{\rho(\boldsymbol{M_{i}}): \boldsymbol{M_{i}} \in M\} \nonumber \text{.}
\end{equation}

\noindent Therefore, the joint spectral radius of set $M$ of symmetric matrices is equal to the largest spectral radius of matrices in the set. Based on Theorem \ref{thm:JSR}, we conclude that if the largest spectral radius of system matrices of an undirected dynamic network is less than $1$, the origin is an asymptotically stable equilibrium point and the epidemic dies out.
\end{proof}

\begin{crl} \label{crl:Static}
Consider a static epidemic network with adjacency matrix $\boldsymbol{A}$, infection probability $\beta$, and recovery probability $\delta$. The epidemic dies out if $\frac{\beta}{\delta}<\frac{1}{\rho(\boldsymbol{A})}$.
\end{crl}

\begin{proof}
For a static network, $M$, the set of system matrices has only one element which is $(1-\delta)\boldsymbol{I}+\beta \boldsymbol{A}$. In this case, $\overline{\rho}_{k}(M)$, the largest possible spectral radius of all products of $k$ matrices chosen in the set $M$, can be written as

\begin{equation}
\overline{\rho}_{k}(M)=sup\left \{\rho\left(\prod_{i=1}^{k}{\boldsymbol{M_{i}}}\right):\boldsymbol{M_{i}}\in M\right\}=\rho((1-\delta)\boldsymbol{I}+\beta \boldsymbol{A})^{k} \nonumber \text{.}
\end{equation}

\noindent $M$ is a bounded set, and \cite{Wang} proves that, for a bounded set of matrices the joint spectral radius is equal to the generalized spectral radius. Hence, we can calculate the joint spectral radius as

\begin{equation}
\widehat{\rho}(M)=\lim_{k\to\infty}\overline{\rho}_{k}(M)^{\frac{1}{k}}= \rho((1-\delta)\boldsymbol{I}+\beta \boldsymbol{A}) \nonumber \text{.}
\end{equation}

\noindent According to Theorem \ref{thm:JSR}, the epidemic dies out if the joint spectral radius of the set of system matrices is less than $1$. For a static network, the joint spectral radius is equal to $1-\delta+\beta \rho(\boldsymbol{A})$. Therefore, the epidemic dies out if

\begin{equation}
\frac{\beta}{\delta}<\frac{1}{\rho(\boldsymbol{A})} \label{eq:static} \text{.}
\end{equation}

\end{proof}
\noindent The epidemic threshold for static networks in Eq. (\ref{eq:static}) is identical to the one in its analytical results in \cite{Faloutsos3}. \\

\begin{crl} \label{crl:Periodic}
Consider a dynamic network with a fixed repetition pattern of T adjacency matrices in a set $L=\left\{\boldsymbol{A_{1}},\boldsymbol{A_{2}},...,\boldsymbol{A_{T}}\right\}$, with infection probability $\beta$ and recovery probability $\delta$. The epidemic dies out if

\begin{equation}
\rho(\prod_{i=1}^{T}{\left[(1-\delta)\boldsymbol{I}+\beta \boldsymbol{A_{i}}\right]})<1 \nonumber \text{.}
\end{equation}

\end{crl}

\begin{proof}
Consider a dynamic network with a fixed repetition pattern of $T$ adjacency matrices and $k=mT$ where $m$ is a positive integer. For the case where $k=mT$, $\overline{\rho}_{k}(M)$ can be written as

\begin{eqnarray}
&&\overline{\rho}_{k}(M)=sup\left \{\rho\left(\prod_{i=1}^{mT}{\boldsymbol{M_{i}}}\right)\right\}\nonumber\\&&=sup\left \{\rho\left(\prod_{i=1}^{m}{\left[\prod_{i=1}^{T}\left[(1-\delta)\boldsymbol{I}+\beta \boldsymbol{A_{i}}\right]\right]}\right)\right\} \nonumber \\&&=\rho\left(\prod_{i=1}^{T}\left[(1-\delta)\boldsymbol{I}+\beta \boldsymbol{A_{i}}\right]\right)^{m}
\text{.}
\end{eqnarray}

\noindent Because the set of system matrices is bounded, its joint spectral radius is equal to its generalized spectral radius. Hence, we can calculate the joint spectral radius as

\begin{equation}
\widehat{\rho}(M)=\lim_{m\to\infty}\overline{\rho}_{mT}(M)^{\frac{1}{mT}}= \rho\left(\prod_{i=1}^{T}\left[(1-\delta)\boldsymbol{I}+\beta \boldsymbol{A_{i}}\right]\right)^{\frac{1}{T}} \nonumber \text{.}
\end{equation}

\noindent According to Theorem \ref{thm:JSR}, if the joint spectral radius of the set of system matrices is less than $1$, the epidemic dies out. Therefore, in this case, the epidemic dies out if $\rho\left(\prod_{i=1}^{T}\left[(1-\delta)\boldsymbol{I}+\beta \boldsymbol{A_{i}}\right]\right)^{\frac{1}{T}}<1$ or equivalently

\begin{equation}
\rho(\prod_{i=1}^{T}{\left[(1-\delta)\boldsymbol{I}+\beta \boldsymbol{A_{i}}\right]})<1 \label{eq:periodic} \text{.}
\end{equation}

\end{proof}

\noindent The derived epidemic threshold for the periodic dynamic network is the same as the threshold in \cite{Faloutsos1}.

\indent In the next corollary, we propose the condition under which epidemics die out in the case of dynamic networks with all the elements of set $L$ corresponding to regular networks. In regular dynamic networks, although links between nodes are dynamic and a given node can change its neighbors at each index, all nodes have the same node degree, and they preserve their node degrees.

\begin{crl} \label{crl:regular}
The epidemic in a dynamic regular network with undirected graphs and node degree of $\overline{k}$ dies out if 
$\frac{\beta}{\delta}<\frac{1}{\overline{k}}$.
\end{crl}

\begin{proof}
We know that the spectral radius of a regular symmetric graph is equal to its node degree. Considering $\boldsymbol{M_{i}}=(1-\delta)\boldsymbol{I}+\beta \boldsymbol{A_{i}}$, we can calculate the spectral radius of $\boldsymbol{M_{i}}$ for any regular adjacency matrix as follows:

\begin{equation}  
\rho(\boldsymbol{M_{i}})=1-\delta+\beta \rho(\boldsymbol{A_{i}})= 1-\delta+\beta \overline{k}\nonumber\text{.}
\end{equation}

\noindent Therefore, all system matrices have the spectral radius of $1-\delta+\beta \overline{k}$, so the largest spectral radius of the system matrices is also equal to $1-\delta+\beta \overline{k}$. Based on Theorem \ref{thm:symmetric}, we can conclude that epidemics in dynamic networks with regular undirected graphs die out if the largest spectral radius of the system matrices is less than $1$. Hence, the epidemic dies out if $1-\delta+\beta \overline{k} <1$ or equivalently if

\begin{equation}
\frac{\beta}{\delta}<\frac{1}{\overline{k}} \label{eq:reg} \text{.}
\end{equation}
\end{proof}

\noindent Eq. (\ref{eq:reg}) remarkably is identical to the epidemic threshold for static regular networks \cite{Faloutsos3}. \\
\indent Hitherto, we have studied dynamic networks with deterministic and given adjacency matrices. However, for cases in which adjacency matrices are stochastic rather than deterministic, the joint spectral radius of set of the system matrices is a random variable. Therefore, the condition under which epidemics die out is expressed in terms of statistical characteristics of the joint spectral radius. 
Dynamic Gilbert network is an example of such dynamic network with stochastic adjacency matrices. In a Gilbert network with parameter $P$, every link exists with probability $P$ \cite{Gilbert}. In other words, the existence of a link is a Bernoulli random variable with parameter $P$. For a dynamic Gilbert network, at every time instant, these binary random variables are redrawn according to Bernoulli distribution. 
In the next corollary, we derive the upper bound for the probability of spreading epidemics for Gilbert dynamic networks in terms of the expected value of the joint spectral radius.

\begin{figure*}[ht]
  \includegraphics[width=\textwidth]{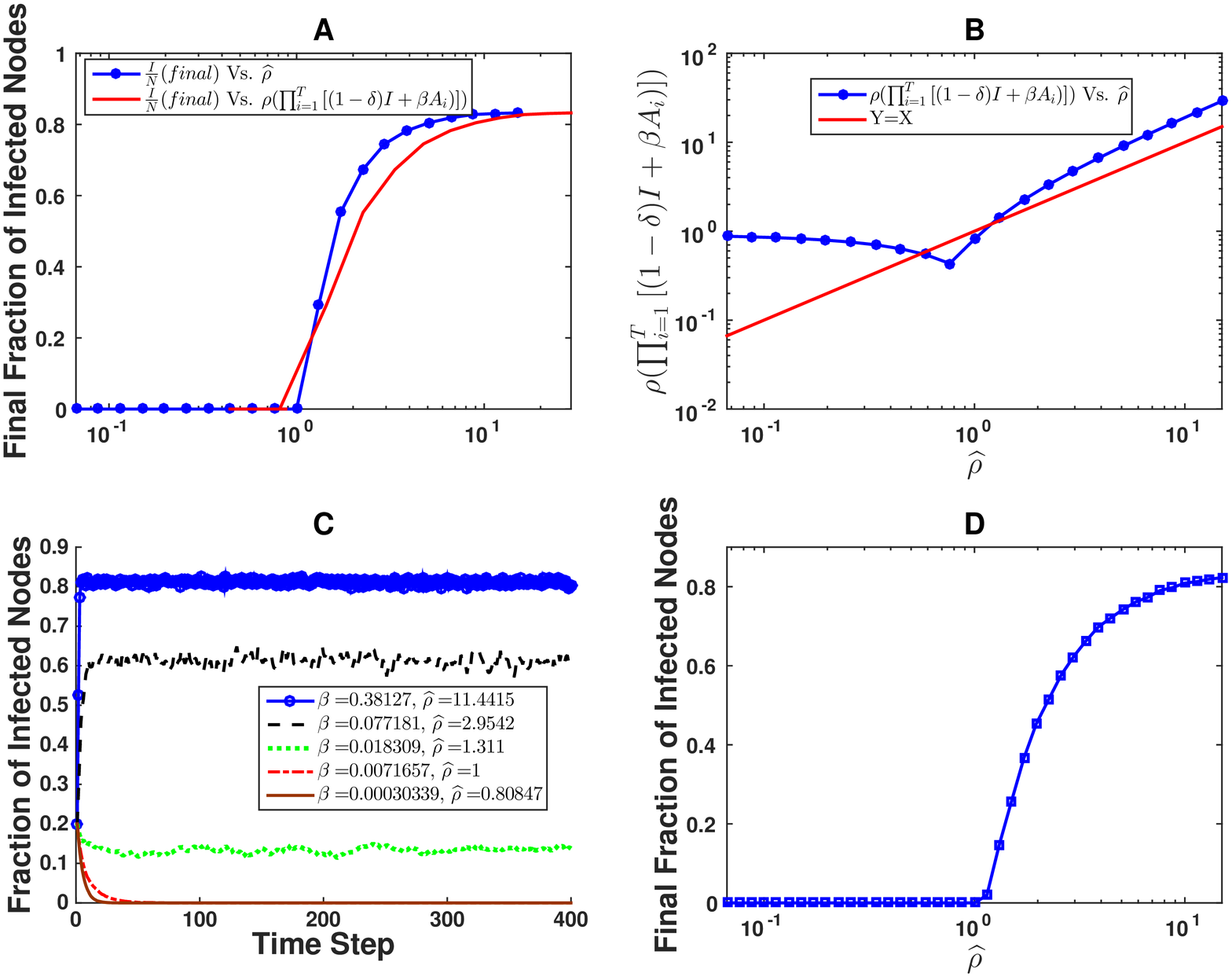}
  \caption{Simulation results of epidemics on Watts-Strogatz and Barabasi-Albert dynamic networks (A) Final fraction of infected nodes for dynamic Watts-Strogatz networks with $1000$ nodes and rewiring probability of $0.5$. (B) Comparison between the joint spectral radius and the spectral radius of the product of system matrices for dynamic Watts-Strogatz networks. (C) Fraction of infected nodes over time for dynamic Barabasi-Albert network with $1000$ nodes. (D) Final fraction of infected nodes for dynamic Barabasi-Albert network with $1000$ nodes vs. the joint spectral radius.}
  \label{fig:Duke1}
\end{figure*}

\begin{crl} \label{crl:Gilbert}
For a given dynamic Gilbert network with $N$ nodes and probability $P$ of the existence of links, the probability that an epidemic spreads is upper bounded by $1-\delta+(N-1)\beta P$.
\end{crl}

\begin{proof}
We define $\boldsymbol{\widehat{M}}$ as $\boldsymbol{\widehat{M}}=\prod_{i=1}^{k}{\boldsymbol{M_{i}}}$ where $\boldsymbol{M_{i}}$ denotes the system matrix corresponding to one realization of a Gilbert dynamic network's adjacency matrix. $\widehat{m}_{q,n}$ denotes the element in the $q^{th}$ row and $n^{th}$ column of matrix $\boldsymbol{\widehat{M}}$. In the Appendix, we show that for all columns of $\boldsymbol{\widehat{M}}$

\begin{equation}
E\left\{\sum_{q=1}^{N}{|\widehat{m}_{q,n}|}\right\}=\left[1-\delta+(N-1)\beta P\right]^k  \forall n=1,2,...,N \label{eq:inf} \text{,}
\end{equation}

\noindent where $E$ denotes the expected value. Considering  $\left|\left|\boldsymbol{\widehat{M}}\right|\right|_{1}=\max_{n}{\sum_{q=1}^{N}{|\widehat{m}_{q,n}|}}$ and (\ref{eq:inf}), we can calculate $E\left\{\left|\left|\boldsymbol{\widehat{M}}\right|\right|_{1}\right\}$ as 
\begin{equation}
E\left\{\left|\left|\boldsymbol{\widehat{M}}\right|\right|_{1}\right\}=\left[1-\delta+(N-1)\beta P\right]^k \nonumber \text{.}
\end{equation}

\noindent Because the above equality holds for any product of $\boldsymbol{M_{i}}$s, $E\left\{\widehat{\rho}_{k}(M)\right\}$ can be written as

\begin{equation} E\left\{\widehat{\rho}_{k}(M)\right\}=\left[1-\delta+(N-1)\beta P\right]^k \nonumber \text{.}
\end{equation}

\noindent Consequently, the expected value of the joint spectral radius can be calculated as 

\begin{equation}
E\left\{\widehat{\rho}(M)\right\}=E\left\{\lim_{k\to\infty}\widehat{\rho}_{k}(M)^{\frac{1}{k}}\right\}=\left[1-\delta+(N-1)\beta P\right] \nonumber\text{.}
\end{equation}

\noindent We employ the Markov inequality and the expected value of the joint spectral radius to compute the upper bound for the probability that the joint spectral radius is more than $1$. Using the Markov inequality, we can write 

\begin{equation}
Prob(\widehat{\rho}(M)\geq 1)\leq E\{\widehat{\rho}(M)\} \label{eq:markovin}\text{.}
\end{equation} 

\noindent Substituting the expected value of the joint spectral radius in Eq. (\ref{eq:markovin}), we conclude that

\begin{equation}
Prob(\widehat{\rho}(M)\geq 1)\leq 1-\delta+(N-1)\beta P \label{eq:markovin2}\text{.}
\end{equation} 

\noindent According to Theorem \ref{thm:JSR}, the epidemic dies out if the joint spectral radius is less than $1$. Therefore, the probability of the epidemic spreading is equal to the probability that the joint spectral radius is greater than $1$. Considering Eq. (\ref{eq:markovin2}), we conclude that the probability of epidemic spread is upper bounded by $1-\delta+(N-1)\beta P$. 
\end{proof}

However, when $[1-\delta+(N-1)\beta P]$ is greater than $1$, this condition is not informative. Therefore, we consider the $min\{1,[1-\delta+(N-1)\beta P]\}$ as the upper bound for the probability of spreading.

\section{SIMULATION RESULTS\label{section:simulation}}

In this section, we validate our theoretical results by simulating an epidemic on synthetic and real dynamic networks. First, we simulate an epidemic on a dynamic Watts-Strogatz network and compare the derived epidemic threshold with the threshold proposed in \cite{Faloutsos2}. Then, the simulation result of the final fraction of infected nodes versus the joint spectral radius for a dynamic Barabasi-Albert network is presented. Further, we evaluate our analytical results for real networks by simulating an epidemic on the set of extracted graphs from the MIT Reality Mining data set \cite{MIT}. Next, an epidemic on a dynamic regular network is simulated. Finally, we validate the derived upper bound for the probability of the epidemic spreading in a dynamic Gilbert network using the simulation results.\\
\begin{figure*}[ht]
  \includegraphics[width=\textwidth]{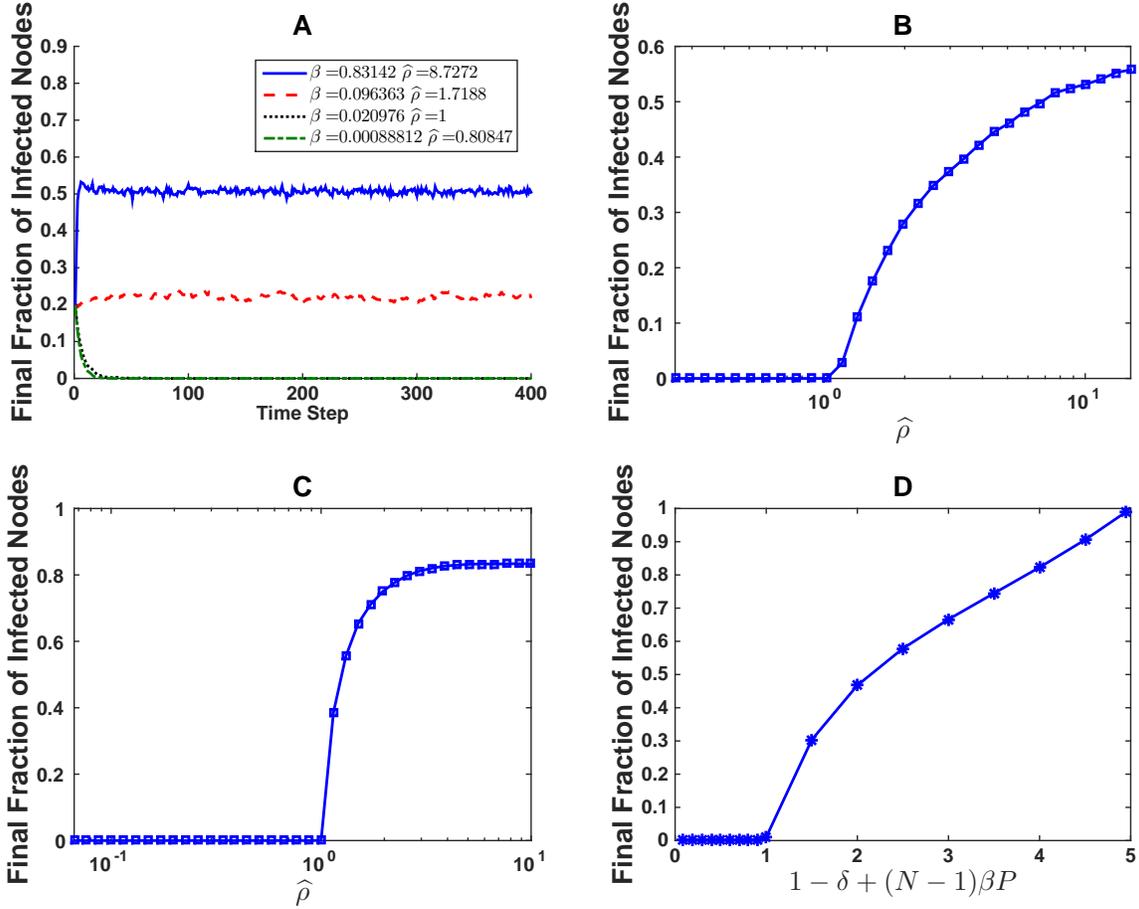}
  \caption{(A) Fraction of infected nodes for the MIT Reality Mining dynamic over time. (B) Final fraction of infected nodes for MIT Reality Mining dynamic vs. the joint spectral radius. (C) Final fraction of infected nodes for a dynamic regular network with $1000$ nodes and node degree of $8$. (D) Final fraction of infected nodes for a dynamic Gilbert network with $1000$ nodes and node degree of $8$. }
    \label{fig:Duke2}
\end{figure*}
\indent Fig. \ref{fig:Duke1} illustrates the simulation results of epidemics on Watts-Strogatz and Barabasi-Albert dynamic networks. Both networks contain $1000$ nodes, and the value of $\delta=0.2$ is kept constant while increasing $\beta$. The final fraction of infected nodes is denoted by $\frac{I}{N}$, where $I$ denotes the final number of infected nodes and $N$ the total number of nodes. Fig. \ref{fig:Duke1} (A) plots $\frac{I}{N}$ versus the joint spectral radius of system matrices and the spectral radius of the products of the system matrices for a dynamic Watts-Strogatz network with rewiring probability $0.5$. In order to realize a dynamic Watts-Strogatz network, the adjacency matrix of the network at each index is chosen randomly from a set of four Watts-Strogatz graphs with average node degrees of $4$, $8$, $12$, and $16$ and, spectral radii, respectively, of $4.46242$, $8.41081$, $12.40911$, and $16.38739$. 
The joint spectral radius of the set of system matrices is equal to the spectral radius of the system matrix corresponding to the adjacency matrix with the largest spectral radius.
We increase $\beta$ from $0.00052$ to $0.86652$ in order to generate different epidemic strengths. The number of iterations for each case is $20$. The observation is made that epidemics die out in all cases in which the value of the joint spectral radius is less than $1$. As soon as the value of the joint spectral radius increases beyond $1$, the epidemic spreads, thus confirming the analytical results of Theorems \ref{thm:JSR} and \ref{thm:symmetric}.    
\noindent The curve of the final fraction of infected nodes versus the spectral radius of the system matrices product shows that  $\rho(\prod_{i=1}^{T}{\left[(1-\delta)\boldsymbol{I}+\beta \boldsymbol{A_{i}})\right]})$ is not an accurate epidemic threshold. Epidemics spread for some of its values that are less than $1$. This simulation result contradicts analytical results of Theorem I in \cite{Faloutsos2} that state that if the spectral radius of the product is less than $1$, the epidemic dies out. \\
\indent Fig. \ref{fig:Duke1} (B) compares the spectral radius of the product of system matrices of the dynamic Watts-Strogatz network with the joint spectral radius of its set of system matrices. The curve of $Y=X$ allows us to determine the values of the joint spectral radius for which the spectral radius of the product is greater than the joint spectral radius or vice versa. Some points in Fig. \ref{fig:Duke1} (B) have joint spectral radius greater than $1$, and the spectral radius of the product less than $1$, resulting in incorrect predictions of epidemics dying out for these cases if we choose the spectral radius of the product as the epidemic threshold. \\
\indent Fig. \ref{fig:Duke1} (C) shows the fraction of infected nodes over time for a dynamic Barabasi-Albert network. In order to realize a dynamic Barabasi-Albert network, we have selected four Barabasi-Albert graphs with average node degree of $4$, $8$, $12$, and $16$ with, spectral radii, respectively, of $12.66217$, $17.36462$, $22.36887$, and $27.91071$, as the set of adjacency matrices. During simulation, the adjacency matrix of this dynamic network at each index is randomly chosen from this set of matrices with equal probability.
The number of iterations for all cases is $20$, and the initial fraction of infected nodes is $0.2$.
Increasing the value of $\beta$ leads to an increase in the joint spectral radius and, eventually, the final fraction of infected nodes. Also, the epidemics die out for cases in which the joint spectral radius is less than $1$. Fig. \ref{fig:Duke1} (D) depicts the final fraction of infected nodes for the dynamic Barabasi-Albert network where $\beta$ increases from $ 0.00038$ to $0.63481$. Results show that epidemics spread for cases in which the joint spectral radius is greater than $1$. \\
\indent Fig. \ref{fig:Duke2} (A) shows the fraction of infected nodes over time for a dynamic network that has time-varying adjacency matrices extracted from the MIT Reality Mining data set \cite{MIT}. This data set contains the adjacency connectivity matrix of $94$ persons, obtained by using mobile phones pre-installed with various specific software, including a logger of Bluetooth devices that was triggered when the distance between two mobile phones was approximately $5$ m or less. Bluetooth scans were carried out every $5$ min. This data set contains information collected from mobile phones from September $2004$ to June $2005$. We extract eight adjacency matrices for $8$ consecutive hr, from $8:00$ a.m. to $4:00$ p.m. on September 1, 2004. Specral radii of these matrices are, respectively, $6.30117$, $5.41546$, $9.44439$, $9.09696$, $8.36535$, $9.53451$, $9.05251$, and $7.41181$. At each time index, the adjacency matrix was randomly chosen from the extracted matrices. For all the four cases, $\delta=0.2$ and the initial fraction of infected nodes is $0.2$. The number of iterations for each case is $50$. Increasing $\beta$ makes the joint spectral radius larger, and epidemic spreads for the values of the joint spectral radius larger than $1$. \\
\indent Fig. \ref{fig:Duke2} (B) depicts the final fraction of infected nodes for the dynamic network in Fig. \ref{fig:Duke2} (A). We fix the value of $\delta$ at $0.2$ and increase the value of $\beta$ from $0.00089$ to $0.83142$. The number of iterations for each case is $50$, therefore, the epidemic spreads when the joint spectral radius is greater than $1$. For cases in which the epidemic spreads, the final fraction of infected nodes increases with increasing value of joint spectral radius. \\
\indent Fig. \ref{fig:Duke2} (C) shows the final fraction of infected nodes for a dynamic regular network with $1000$ nodes and node degree $\overline{k}=8$ versus the product of $\frac{\beta}{\delta}$ and node degree. 
We simulate epidemics for various values of $\beta$ in the interval $[0.00106,0.99089]$, while fixing the value of $\delta$ at $0.2$. The number of iterations for each case is $100$. Therefore, epidemics die out when $\frac{\beta}{\delta}<\frac{1}{\overline{k}}$, confirming the result of Corollary \ref{crl:regular}. \\
\begin{table}[h]
\caption{$\delta$ and $\beta$ values used in simulation of the epidemic in the dynamic Gilbert network}
\label{tab:sim}
\begin{center}
\begin{tabular}{c|c|c}
$\boldsymbol{\delta}$ & $\boldsymbol{\beta}$ & $\boldsymbol{1-\delta+(N-1)\beta P}$\\
\hline
$0.95$ & $0.01$ & $0.09$\\
$0.85$ & $0.01$ & $0.19$\\
$0.74$ & $0.01$ & $0.3$\\
$0.64$ & $0.01$ & $0.4$\\
$0.54$ & $0.01$ & $0.5$\\
$0.44$ & $0.01$ & $0.6$\\
$0.34$ & $0.01$ & $0.7$\\
$0.24$ & $0.01$ & $0.8$\\
$0.14$ & $0.01$ & $0.9$\\
$0.04$ & $0.01$ & $1$\\
$0.7$ & $0.3$ & $1.5$\\
$0.6$ & $0.4$ & $2$\\
$0.5$ & $0.5$ & $2.5$\\
$0.4$ & $0.6$ & $3$\\
$0.3$ & $0.7$ & $3.5$\\
$0.2$ & $0.8$ & $4$\\
$0.1$ & $0.9$ & $4.5$\\
$0.01$ & $0.99$ & $4.95$\\
\end{tabular}
\end{center}
\label{tbl:parameter}
\end{table}

Fig. \ref{fig:Duke2} (D) depicts the final fraction of infected nodes for a dynamic Gilbert network with $1000$ nodes and a probability of connection $P=0.004$ versus $1-\delta+(N-1)\beta P$, which is the upper bound for the probability of epidemic spread derived in Corollary \ref{crl:Gilbert}.
Table \ref {tab:sim} shows the chosen values of $\delta$ and $\beta$ in the simulation as well as values of $1-\delta+(N-1)\beta P$. The epidemic dies out up to the point at which the upper bound is less than $1$. When this upper bound reaches $1$, the epidemic begins to spread, confirming the result of Corollary \ref{crl:Gilbert}. Although having an upper bound greater than $1$ for a probability is not informative, it can be used as a measure of epidemic strength. This just described situation in Fig. \ref{fig:Duke2} (D), as increasing the value of the upper bound leads to an increase in the final fraction of infected nodes.

\section{CONCLUSION}
In this paper, we study the spread of SIS epidemics in dynamic networks. We propose an approach to derive the analytical epidemic threshold that can be applied to any dynamic network with an adjacency matrix randomly chosen from a set of matrices at each index. A linearized version of the nonlinear epidemic system is employed to derive the epidemic threshold. We show that an epidemic dies out if the origin is an asymptotically stable equilibrium point. We derive the epidemic threshold for dynamic networks using the joint spectral radius of  system matrices. We calculate the simplified version of the epidemic threshold for undirected dynamic networks based on the fact that the joint spectral radius of a set of symmetric matrices is equal to the largest spectral radius of the matrices in that set. Then, we derive the epidemic threshold for dynamic regular networks. For dynamic Gilbert networks, we compute the upper bound of the probability of an epidemic spreading in terms of the expected value of the joint spectral radius.\\
\indent Our analytical results show that epidemic thresholds of dynamic networks are determined by the dynamic topologies and epidemic strengths in the networks. In particular, the joint spectral radius of the set of system matrices determines whether or not the epidemic dies out. In other words, the joint spectral radius characterizes epidemic strengths and the level of connectivity between the nodes in dynamic networks over time. In the case of undirected networks, the joint spectral radius is dependent only on the adjacency matrix with the largest spectral radius, thereby implying that, for undirected networks, dynamics of epidemics are determined primarily by the adjacency matrix with the largest spectral radius. Variance in topology of dynamic networks impacts the spread of epidemics. In order to validate our analytical results for dynamic networks with various dynamic topologies, we simulate epidemics on Watts-Strogatz, Barabasi-Albert, Regular, MIT Reality Mining, and Gilbert dynamic networks. The dynamic Watts-Strogatz network is used to model a dynamic network with small-world properties, including a high clustering coefficient and small shortest path. On the other hand, the dynamic Gilbert network is employed to validate the derived epidemic threshold for dynamic networks with random structure and small clustering coefficients. Many networks in the real world, including e-mail networks, the world wide web, and biological networks, are considered scale-free networks that can be modeled by Barabasi-Albert networks \cite{Faloutsos5, Barabasi, Ebel, Clauset,Medina,Vespignani}. Our simulation result for dynamic Barabasi-Albert networks validates the derived epidemic threshold for such scale-free networks. Also, we verify our theoretical result with a real-life network. We simulate epidemics on the dynamic network with adjacency matrices extracted from the MIT Reality Mining data set \cite{MIT}, proving that the derived epidemic threshold holds for this real-life network as well. \\
\indent The model of dynamic switching networks in this paper does not take into account any temporal correlation between consecutive adjacency matrices. This assumption does not hold for dynamic networks with strong correlation between consecutive adjacency matrices. For example, for the simulation of epidemics on the MIT Reality Mining data set, the adjacency matrix at each time step is chosen randomly and independently from other time steps. However, in reality, the sequence of adjacency matrices follows a temporal order, implying the existence of correlation between the adjacency matrices.
The derived epidemic threshold corresponds to the most possible vulnerable sequence of matrices against epidemics, because in our framework, we assume that adjacency matrices are chosen completely at random and independent of each other. 
That is because the joint spectral radius of the set of system matrices is determined by the sequence of adjacency matrices that are the most vulnerable to epidemics. Thus, in the case of dynamic networks with temporal correlation, if the framework predicts the death of the epidemic, the epidemic will die out, even though this model ignores temporal correlations. However, the prediction can be too conservative; That is, it may predict the spread of epidemic, when the epidemic actually dies out.\\
\indent Our analysis is based on the assumptions that the parameters and probabilities of the epidemics, the underlying structure, and topology of the transmission network are known. However, these assumptions will not hold in all scenarios. A robustness analysis with respect to model parameters can validate the application of the framework to such scenarios. We illustrate, all through this paper, some aspects of the robustness of our analysis by choosing adjacency matrices randomly and selecting widely ranging values of epidemic probabilities in the simulations. The framework developed in this paper may be extended to the derivation of epidemic thresholds for other types of epidemic models. It can be applied to cases in which the linearized version of the dynamical system corresponding to the epidemic model at the disease-free equilibrium point can be expressed in a matrix form similar to Eq. (\ref{eq:PM}).  


\appendix

 \label{ap:final}
\begin{thm}\label{thm:terkoondam}Consider $\boldsymbol{\widehat{M}}=\prod_{i=1}^{k}{\boldsymbol{M_{i}}}$ where $\boldsymbol{M_{i}}$ denotes the system matrix corresponding to one realization of a Gilbert dynamic network's adjacency matrix. For the matrix $\boldsymbol{\widehat{M}}$, the expected summation value of each column's elements is 

\begin{equation}
E\left\{\sum_{q=1}^{N}{|\widehat{m}_{q,n}|}\right\}=\left[1-\delta+(N-1)\beta P\right]^k  \forall n=1,2,...,N \label{eq:terkoondam} \text{,}
\end{equation}

\noindent where $P$ denotes the probability of link existence and $\widehat{m}_{q,n}$ denotes the element in $q^{th}$ row and $n^{th}$ column of matrix $\boldsymbol{\widehat{M}}$.
\end{thm}

\begin{proof}
We prove this theorem via induction. In the case of a Gilbert dynamic network, off-diagonal elements of adjacency matrix $\boldsymbol{A}$ are independent and identically distributed (iid) Bernoulli random variables with parameter $P$. The first step is to show that Eq. (\ref{eq:terkoondam}) is correct when $k=1$. Assume $k=1$. In this case, $\boldsymbol{\widehat{M}}=(1-\delta)\boldsymbol{I}+\beta \boldsymbol{A}$ and $E\left\{\sum_{q=1}^{N}{|\widehat{m}_{q,n}|}\right\}$ can be written as 

\begin{equation}
E\left\{\sum_{q=1}^{N}{|\widehat{m}_{q,n}|}\right\}=(1-\delta)+\beta\sum_{i=1}^{N-1}{E\{X_{i}\}} \label{eq:t1} \text{,}
\end{equation}

\noindent where the $X_{i}$'s are iid random variables with parameter $P$. $E\{X_{i}\}=P$. Therefore, we can rewrite Eq. (\ref{eq:t1}) as

\begin{equation}
E\left\{\sum_{q=1}^{N}{|\widehat{m}_{q,n}|}\right\}=(1-\delta)+(N-1)\beta P \label{eq:t2} \text{.}
\end{equation}

\noindent Eq. (\ref{eq:t2}) shows that Eq. (\ref{eq:terkoondam}) is correct for $k=1$. The second step is to assume that Eq. (\ref{eq:terkoondam}) is correct for $k$ and prove it for $k+1$. Assume $\boldsymbol{\widehat{M}}=\prod_{i=1}^{k}{\boldsymbol{M_{i}}}$, considering the assumption of correctness of Eq. (\ref{eq:terkoondam}) for $k$, $E\left\{\sum_{q=1}^{N}{|\widehat{m}_{q,n}|}\right\}=\left[1-\delta+(N-1)\beta P\right]^k $. Suppose $R=M_{k+1}\boldsymbol{\widehat{M}}$. $r_{q,n}$, the element in the the $q^{th}$ row and the $n^{th}$ column of $R$ can be written in terms of the elements of $\boldsymbol{\widehat{M}}$ as

\begin{equation}
r_{q,n}=(1-\delta)\widehat{m}_{q,n}+\beta\sum_{j=1, j\neq q}^{N} {X_{j}\widehat{m}_{j,n}} \text{,}
\end{equation}

\noindent where the $X_{j}$'s are iid Bernoulli random variables with parameter $P$. Therefore, we can write $\sum_{q=1}^{N}{|r_{q,n}|}$ as 

\begin{equation}
\sum_{q=1}^{N}{|r_{q,n}|}=(1-\delta)\sum_{q=1}^{N}{\widehat{m}_{q,n}}+\beta\sum_{q=1}^{N}{\left[\sum_{j=1, j\neq q}^{N} {X_{j}\widehat{m}_{j,n}}\right]}\label{eq:t3}\text{,}
\end{equation}

\noindent where $X_{j}$'s and $\widehat{m}_{j,n}$ are independent. Hence, $E\left\{\sum_{q=1}^{N}{|r_{q,n}|}\right\}$ can be written as

\begin{eqnarray}
&&E\left\{\sum_{q=1}^{N}{|r_{q,n}|}\right\}=(1-\delta)E\left\{\sum_{q=1}^{N}{\widehat{m}_{q,n}}\right\}+\nonumber\\&&P\beta E\left\{\sum_{q=1}^{N}{\left[\sum_{j=1,j\neq q}^{N} {\widehat{m}_{j,n}}\right]}\right\}\label{eq:t4}\text{.}
\end{eqnarray}

\noindent However, $E\left\{\sum_{q=1}^{N}{\left[\sum_{j=1,j\neq q}^{N} {\widehat{m}_{j,n}}\right]}\right\}=(N-1)E\left\{\sum_{q=1}^{N}{\widehat{m}_{q,n}}\right\}$. Considering $E\left\{\sum_{q=1}^{N}{\widehat{m}_{q,n}}\right\}=\left[(1-\delta)+\beta P (N-1)\right]^{k}$, we can rewrite (\ref{eq:t4}) as

\begin{eqnarray}
&&E\left\{\sum_{q=1}^{N}{|r_{q,n}|}\right\}=\left[(1-\delta)+\beta P (N-1)\right]E\left\{\sum_{q=1}^{N}{\widehat{m}_{q,n}}\right\}\nonumber\\&&=\left[(1-\delta)+\beta P (N-1)\right]^{k+1}\label{eq:t5}\text{.}
\end{eqnarray}

\noindent The result in Eq. (\ref{eq:t5}) for $k+1$ is the last step in the proof of this theorem through induction.

\end{proof}


\begin{thebibliography}{1}

\bibitem{Kostova} 
T. Kostova, ``Interplay of node connectivity and epidemic rates in the dynamics of epidemic networks", Journal of Difference Equations and its Applications, vol. 15, pp. 415-428, 2009.

\bibitem{Faloutsos3} 
D. Chakrabarti, Y. Wang, C. Wang, J. Leskovec, and C. Faloutsos, "Epidemic thresholds in real networks". ACM Trans. Inf. Syst. Secur. v. 10, pp. 1-26, 2008.
 
\bibitem{Faloutsos4} 
Yang Wang, D. Chakrabarti, Chenxi Wang, and C. Faloutsos, ``Epidemic spreading in real networks: an eigenvalue viewpoint," In proceedings of the 22nd International Symposium on Reliable Distributed Systems, pp. 25- 34, 2003.

\bibitem{Chen} 
Z. Chen and C. Ji, ``Spatial-temporal modeling of malware propagation in networks", IEEE Transactions on Neural Networks, vol. 16, no. 5, pp. 1291-1303, 2005.

\bibitem{Piet} 
P. V. Mieghem, J. Omic, and R. Kooij, ``Virus Spread in Networks", IEEE/ACM Transactions on Networking, vol. 17, no. 1, pp. 1-14, Feb. 2009.

\bibitem{Ganesh} 
A. Ganesh, L. Massoulie, and D. Towsley, "The effect of network topology on the spread of epidemics," In Proceedings IEEE INFOCOM, v. 2, pp. 1455-1466, 2005.

\bibitem{Lindquist} 
J. Lindquist, J. Ma, P. V. D. Driessche, F. H. Willeboordse, ``Effective degree network disease models", Journal of Mathematical Biology, v. 62, no. 2, pp. 143-164, 2010.  

\bibitem{Vespignani2} 
M. Bogua, R. Pastor-Satorras, and A. Vespignani, ``Epidemic spreading in complex networks with degree correlations," in Statist. Mech. Complex Netw., pp. 127-147, 2003.

\bibitem{Faloutsos1} 
B. A. Prakash, H. Tong, N. Valler, M. Faloutsos, and C. Faloutsos, ``Virus propagation on time-varying networks: theory and immunization algorithms," Machine learning and knowledge discovery in databases, v. 6323, pp. 99-114, 2010.

\bibitem{Faloutsos2} 
N. C. Valler, B. A. Prakash, H. Tong, M. Faloutsos, and C. Faloutsos, ``Epidemic Spread in Mobile Ad Hoc Networks: Determining the Tipping Point," In Proceedings of the 10th international IFIP TC 6 conference on Networking, Volume Part I, pp. 266-280, 2011.

\bibitem{Volz} 
E. Volz and L. A. Meyers, ``Epidemic thresholds in dynamic contact networks", Journal of the Royal Society Interface, v. 6, pp. 233-241, 2009.

\bibitem{Taylor} 
M. Taylor, T. J. Taylor, and I. Z. Kiss, ``Epidemic threshold and control in a dynamic network", Phys. Rev. E., v. 85, pp. 016103,  2012.

\bibitem{Gross} 
T. Gross, C. J. Dommar D'Lima, and B. Blasius, ``Epidemic dynamics in an adaptive network", Phys. Rev. Lett., v. 96, pp. 208701, 2006.


\bibitem{Rocha} 
L. E. C. Rocha and V. D. Blondel, ``Bursts of vertex activation and epidemics in evolving networks," PLoS Comput. Biol., v. 9, pp. e1002974, 2013.


\bibitem{Bansal} 
S. Bansal, J. Read, B. Pourbohloul, and L. A. Meyers,  ``The dynamic nature of contact networks in infectious disease epidemiology," Journal of biological dynamics, 4 (5), pp. 478-489, 2010. 

\bibitem{Smieszek} 
T. Smieszek, L. Fiebig, and R. W. Scholz, ``Models of epidemics: when contact repetition and clustering should be included," Theoretical Biology and Medical Modeling, vol. 6, 2009.

\bibitem{Kamp} 
C. Kamp, ``Untangling the Interplay between Epidemic Spread and Transmission Network Dynamics," PLoS Comput. Biol., vol. 6, pp. e1000984, 2010.

\bibitem{Rocha2} 
L. E. C. Rocha,  F. Liljeros, and P. Holme, ``Simulated Epidemics in an Empirical Spatiotemporal Network of 50,185 Sexual Contacts," PLoS Comput. Biol., vol. 7, pp. e1001109, 2011.

\bibitem{Volz2} 
E. Volz, and L. A. Meyers, ``Susceptible{\textendash}infected{\textendash}recovered epidemics in dynamic contact networks," Proceedings of the Royal Society of London B: Biological Sciences, vol. 274, pp. 2925-2934, 2007.

\bibitem{Chen5} 
M. I. Chen, A. C. Ghani, and J. Edmunds, ``Mind the Gap: The Role of Time Between Sex With Two Consecutive Partners on the Transmission Dynamics of Gonorrhea," Sexually Transmitted Diseases, vol. 35, pp. 435-444, 2008.

\bibitem{karen1}M. Moslonka-Lefebvre, A. Finley, I. Dorigatti, K. Dehnen-Schmutz, T. Harwood, M.J. Jeger, X.M Xu, O. Holdenrieder, and M. Pautasso, ``Networks in plant epidemiology: From genes to landscapes, countries, and continents,'' Phytopathology, 101(4): 392-403, 2011.

\bibitem{karen2}M.R. Sanatkar, C. Scoglio, B. Natarajan, S. Isard, and K.A. Garrett, ``History, epidemic evolution, and model burn-in for a network of annual invasion: Soybean rust,'' Phytopathology, DOI: 10.1094/PHYTO-12-14-0353-FI.

\bibitem{karen3}M.W. Shaw, and M. Pautasso, ``Networks and plant disease management: concepts and applications,'' Annual Review of Phytopathology, 52(1): 477-493, 2014.

\bibitem{karen4}S. Sutrave, C. Scoglio, S.A. Isard, J.M.S Hutchinson, and K.A Garrett, ``Identifying highly connected counties compensates for resource limitations when sampling national spread of an invasive pathogen,'' PLoS ONE, 7(6): e37793, 2012.

\bibitem{Strang} 
G. C. Rota and G. Strang, ``A note on the joint spectral radius," Proceedings of the Netherlands Academy, v. 22, pp. 379-381, 1995.

\bibitem{Lagarias} 
I. Daubechies and J. C. Lagarias, ``Sets of matrices all infinite products of which converge," Linear Algebra  its Applications, v. 161, pp. 227-263, 1992.

\bibitem{Wang} 
M. A. Berger and Y. Wang, ``Bounded semigroups of matrices," Linear Algebra and its Applications, v. 166, pp. 21-27, 1992.

\bibitem{Lagarias2} 
I. Daubechies and J. C. Lagarias, ``Corrigendum/addendum to: Sets of matrices all infinite products of which converge," Linear Algebra  and its Applications, vol. 327, pp. 69-83, 2001.


\bibitem{MIT} 
N. Eagle, A. Pentland, ``Reality mining: sensing complex social systems,"  Journal Personal and Ubiquitous Computing, v. 10, pp. 255-268, 2006.

\bibitem{Gilbert} 
E. N. Gilbert, ``Random Graphs," in Ann. Math. Statist., vol. 30, pp. 1141-1144, 1959.

\bibitem{Erdos} 
P. Erdos and A. Renyi, ``On the evolution of random graphs," in Publ.
Math. Inst. Hung. Acad. Sci., vol. 5, pp. 17-61, 1960.

\bibitem{Faloutsos5} 
M. Faloutsos, P. Faloutsos, and C. Faloutsos, ``On power-law relationships
of the internet topology," in Proc. ACM SIGCOMM99, pp. 251-262, 1999.

\bibitem{Barabasi} 
A. Barabasi and R. Albert, ``Emergence of scaling in random networks,"
Science 286, pp. 509-512, 1999.

\bibitem{Ebel} H. Ebel, L. Mielsch, and S. Bornholdt, ``Scale-free topology of e-mail networks," in Phys. Rev. E., v. 66, pp. 035-103(R), 2002.

\bibitem{Clauset} 
A. Clauset, Aaron, C. R. Shalizi, and M. E. J. Newman, ``Power-Law Distributions in Empirical Data," SIAM Rev., v. 51,  pp. 661-703, 2009.

\bibitem{Medina} 
A. Medina, A. Lakhina, I. Matta, and J. Byers, ``BRITE: Universal topology generation from a user`s perspective,", Tech. Rep., Boston University, Boston, MA, 2001.

\bibitem{Vespignani} 
A. Vazquez, R. Pastor-Satorras, and A. Vespignani, ``Large-scale topological and dynamical properties of the Internet," Phys. Rev. E., v. 65,  pp. 066-130, 2002.



\end{thebibliography}
\end{document}